%% file: sublinear-hypergraph.tex
\documentclass[11pt]{article}
\input{preamble.tex}

\usepackage{parskip}

\newtheorem{result}{Result}

\newcommand{\davg}{d_{\mathsf{avg}}}
\newcommand{\eavg}{\mu}

\newcommand{\RV}{\mathbf{RV}}
\newcommand{\DEG}{\mathbf{Deg}}
\newcommand{\IncE}{\mathbf{IncEdg}}

\newcommand{\RE}{\mathbf{RE}}
\newcommand{\SIZE}{\mathbf{Size}}
\newcommand{\IncV}{\mathbf{IncVert}}

\newcommand{\est}{\mathsf{est}}

\newcommand{\estavg}{{\tt EstAvgDeg}}
\newcommand{\estH}{{\tt EstH}}
\newcommand{\estmu}{{\tt EstArity}}
\newcommand{\estHyp}{{\tt EstNumHyperEdge}}

\title{Sublinear Algorithms for Estimating the 
	Number of \\ Hyperedges in Arbitrary Hypergraphs\footnote{Preliminary results of this paper were reported in the second author's undergraduate thesis~\cite{LaPorte26}}
	}

\author{Deeparnab Chakrabarty\thanks{Supported by NSF awards CCF-2041920, 2402571.} \\
Dartmouth\\
{\small \tt deeparnab@dartmouth.edu}
\and Cooper LaPorte\\
Dartmouth\\
{\small \tt cooper.h.laporte.26@dartmouth.edu}
\and
C. Seshadhri\thanks{Supported by NSF DMS-2023495, CCF-1740850, 2402572.} \\
University of California, Santa Cruz\\
{\small \tt sesh@ucsc.edu}
}
\date{}
\begin{document}
\maketitle

\begin{abstract}
We study the problem of estimating the number of hyperedges in an arbitrary $n$-vertex hypergraph 
using sublinear in $n$ queries. Note that the number of hyperedges, $m$, can be exponential in $n$. 
For $k$-uniform hypergraphs, estimating $m$ is equivalent to estimating the average vertex degree, a problem studied in Barhum's Master's thesis (Weizmann Inst., 2007) under the standard access model of sampling random vertices, querying vertex degrees, and accessing incident hyperedges. Barhum's techniques do not extend to arbitrary hypergraphs, and simple lower-bound examples show that the standard access model cannot yield strongly sublinear algorithms when hyperedges have unbounded size.

To obtain non-trivial sublinear bounds, we consider a natural generalization of the access
model called the \emph{dual access model}, which allows sampling (labels of) random hyperedges, querying edge sizes, and 
accessing vertices in a hyperedge.
In this model, we give a randomized algorithm that returns a $(1+\varepsilon)$-approximation to $m$ with high probability, making $O(\varepsilon^{-2}\sqrt{n} + \sqrt{n}\log n)$ 
queries. Complementing our algorithm,
we prove a nearly matching lower bound showing that $\Omega(\sqrt{n})$ queries
are necessary for any algorithm that obtains a constant factor approximation to $m$.

\end{abstract}

\thispagestyle{empty}
\newpage
\setcounter{page}{1}
\section{Introduction}

The study of sublinear graph algorithms has been a fertile research area over the past two decades. 
The standard ``adjacency list" query model to access an input graph $G$, arising from the works of Parnas and Ron~\cite{PaRo02}, and Goldreich and Ron~\cite{GR08}, 
is the following: one may (i) sample uniform at random (uar) vertices
in $G$, (ii) perform degree queries, that return the degree of a vertex $v$, and (iii) perform neighbor queries,
that return a uar edge incident to a vertex $v$. 
The most classic (and basic)
algorithmic problem is to estimate the number of edges, $m$, of a graph in sublinear time. 
The first results were by Feige~\cite{Fe06} and by Goldreich and Ron~\cite{GR08}; the latter
give the first sublinear query $(1+\eps)$-approximation algorithm.
Eden, Ron, and Seshadhri~\cite{ERS19} (see also~\cite{BCS26,EdSe26}) gave alternate elegant
$(1+\eps)$-approximation algorithms to the number of edges in $O(\eps^{-2}\sqrt{n/\davg})$ queries,
where $n$ is the number of vertices and $\davg$ is the average degree of $G$.
The problem of estimating $m$ 
has been studied extensively under richer and alternate query models~\cite{EdRo18,EdNaTe22,EdRoRo19,TT22,BT24,BCS26,Ch26}.

Sublinear algorithms for arbitrary graphs have impact beyond theory, and a number
of practical algorithms for analyzing large real-world graphs have used these techniques~\cite{LoBaGo+14,EdJa+18,BeSe20,WaWe23,Wa24}.
Recently, many works in network science and data mining have focused on 
real-world \emph{hypergraphs}~\cite{VeBeKl20,VeBeKl22,AnCoPo+23,CaDeFa+23} (the recent survey~\cite{LeBuEl+25} has \emph{hundreds} of papers on this topic).
A hypergraph $H = (V, E)$ is an arbitrary set system, where each $e \in E$ is a subset of $V$ instead of being a pair.
A common example of a hypergraph is a coauthor network, where the vertex set consists of researchers
and a hyperedge is a subset of coauthors on a publication.
Indeed, many real-world graph datasets are derived from hypergraphs by simply ``collapsing" hyperedges into cliques.
However, hypergraphs are a better model of the original data source, and
 practitioners suggest they should be analyzed directly. This has spurred  practical and theoretical work on hypergraph algorithms~\cite{HwTaTi+08,YuTaWa12,BeAbSc+18,HuQiSh+10,AmVeBe20}
(consider the recent software HyperNetworkX~\cite{hypnx}).

In this paper, we study sublinear algorithms for hypergraphs, and in particular, focus on the classic problem of estimating $m$,
the number of hyperedges. For $k$-uniform hypergraphs\footnote{A hypergraph is $k$-uniform if all hyperedges have exactly $k$ vertices. A standard graph is $2$-uniform.}, estimating $m$ is equivalent to estimating the average degree $\davg$. 
This latter problem was studied by Barhum~\cite{Ba07} in the standard access model described above.
There is a simple lower bound of $\Omega(n^{1-1/k})$ in this model: suppose the hypergraph
has a set of $c n^{1/k}$ vertices ($c$ is a constant) where all subsets of size $k$ form hyperedges,
and the remaining vertices are connected by a $k$-uniform matching. Any constant factor approximation for $m$
requires detecting the set of $\Theta(n^{1/k})$ vertices, and this requires $\Omega(n^{1-1/k})$ random
vertices. Barhum gave a nearly matching upper bound for estimating the average degree in the \emph{bounded arity} case, where all hyperedges have at most $k$ vertices~\cite{Ba07}.

Although bounded arity hypergraphs are an important subclass of hypergraphs, this is not a valid assumption for many real-world hypergraphs.
Our focus is on arbitrary hypergraphs with no assumptions, and given the lower bound mentioned above, 
the problem of estimating $m$ becomes more interesting in this scenario. Note that $m$ can be as large as $\Omega(2^n)$.
Our driving question is: \begin{quote}\emph{For some natural query model for accessing hypergraphs, can we get sublinear, in $n$, algorithms for estimating $m$, the number of hyperedges?} 
\end{quote}
We strive for {\em simple} algorithms that could be implemented in practice, and whose complexity involve as few 
polylogarithmic factors that may arise from ideas such as bucketing.

\paragraph{Our Model and Results.}
Hyperedges and vertices have a natural duality captured by the following  bipartite hyperedge
incidence graph. Consider all elements in $V(H)$ and $E(H)$ represented
by unique IDs. Form the bipartite graph $G = (V(H) \cup E(H), F)$, where the ``left" side
has nodes for $V(H)$, and the ``right" side has nodes for $E(H)$. There is an edge $(v,e)$
for $v \in V(H)$ and $e \in E(H)$ if $v \in e$. We consider the standard graph access model on this  bipartite incidence graph, 
with the ability to sample uniformly at random from each side separately.
In terms of access to the hypergraph, this setup provides the following additional queries:
(iv) sample uar hyperedges (v) arity queries, that given (the ID of) hyperedge $e$, returns the arity $|e|$,
and (vi) uar incidence queries, that given hyperedge $e$, returns a uar vertex $v \in e$. 
We will refer to this model as the \emph{dual access model}. Versions of this model have been discussed
for  sublinear algorithms for the set cover problem~\cite{NO08, YYI12, InMaRu+18}.

\begin{mdframed}[backgroundcolor=gray!20,topline=false,bottomline=false,leftline=false,rightline=false]
\medskip
\begin{result} \label{thm:main}
    Consider dual access to input hypergraphs, where the number of vertices $n$ is known. There
    is a randomized algorithm with the following guarantee. Given an input hypergraph with $m$
    edges and parameter $\eps \in (0,1)$,
    the algorithm outputs a $(1+\eps)$-approximation to $m$ with high probability, and makes $O(\eps^{-2} \sqrt{n} + \sqrt{n}\log n)$ queries. (\Cref{thm:esthyp} in~\Cref{sec:wrapup}.)

    Moreover, any algorithm that outputs a constant factor approximation to $m$ must make $\Omega(\sqrt{n})$ queries. (\Cref{thm:lb-est-m} in~\Cref{sec:lb-est-m}.)
\end{result}
\end{mdframed}
We emphasize that the ability to sample a uniformly random hyperedge does not reveal $m$: the oracle  returns a label drawn uniformly from the unknown set $E(H)$. With only such samples, estimating $m$ is the classical support-size estimation problem and requires $\Theta(\sqrt{m})$ samples~\cite{RoTs16}.
Our algorithm exploits the dual incidence access to replace this dependence on $m$ by one that is sublinear in $n$. Unlike the $k$-uniform case, in an arbitrary hypergraph, the sum of degrees is not a scaled version of $m$, and $m$ can be exponentially large in $n$. So even an $O(n)$-query estimation algorithm is not trivial (although it is fairly simple). 
Since the sum of degrees and number of edges do not satisfy a trivial relationship, estimating the average degree is a different problem.
One of the 
steps in our algorithm solves this problem as well. \smallskip

\begin{mdframed}[backgroundcolor=gray!20,topline=false,bottomline=false,leftline=false,rightline=false]
\smallskip
\begin{result} \label{thm:main-avgdeg}
    Consider the same setup as \Cref{thm:main}. There is a randomized algorithm that outputs a $(1+\eps)$-approximation
    to the average degree whp, and makes $O(\eps^{-1}\sqrt{n\log n})$ queries.    (\Cref{thm:estavg} in~\Cref{sec:est-avg}.)

    Moreover, any algorithm that outputs a constant factor approximation to the average degree requires $\Omega(\sqrt{n})$ queries.
       (\Cref{thm:lb-est-avgd} in~\Cref{sec:lb-est-avgd}.)
\end{result}
\end{mdframed}

\subsection{Main ideas}\label{sec:ideas}
We explain some of the key ideas behind our work, and emphasize the overall
simplicity of the algorithm and its analysis. We begin with some notation.
For a vertex $v$, let $\deg(v)$ be the number of hyperedges incident to $v$,
so the average degree is $\davg := \frac{1}{n} \sum_{v\in V}\deg(v)$. Analogously, we can define the dual
average arity as $\eavg := \frac{1}{m} \sum_{e\in E}|e|$. The equivalent of the standard handshake lemma for
graphs is the equation $n \davg = m \eavg$.

The first insight is considering the {\em fractional degree} $h(x)$, defined as $\sum_{e\ni x} 1/|e|$. 
This was introduced in the context of hypergraphs by Huang, Gleich, and Veldt~\cite{HuGlVe24}, but the concept of ``fractional credit or participation'' has been used in many places, including the website \href{https://csrankings.org/}{https://csrankings.org/}.
The fractional degree has two features: (i) $\sum_{x\in V} h(x) = m$, and
(ii) sampling a random hyperedge $e$ and sampling a vertex uniformly at random from $e$ samples $x$ proportional to $h(x)$. 
Observation (i) implies the natural starting point: if we sample a set $S$ of vertices {\em uniformly} at random
and define $h(S) := \sum_{s \in S} h(s)$, then $\frac{h(S)}{|S|}$ is an unbiased estimate of $m/n$. The variance of $h(x)$ values will create
a problem, but now observation (ii) comes in. By sampling a large enough set $U$ of vertices proportional to $h(x)$,
we can hit all vertices with large enough $h(x)$ value. Overall, we estimate $m$ by obtaining estimates of $h(U)$ and $h(S)$ and returning 
$\widehat{h}(U) + \frac{n}{|S|}\widehat{h}(S)$.
A simple calculation shows that it suffices for $|S|$ and $|U|$ to be $\Ot_\eps(\sqrt{n})$.

However, the exact evaluation of $h(x)$ for any $x\in V$ is too
expensive, since it takes $O(\deg(x))$ queries.
Our next step is to estimate $h(S)$ for a subset $S$ of $\Ot_\eps(\sqrt{n})$ vertices. 
Observing that $h(S) = \sum_{s \in S} \sum_{e \ni s} 1/|e|$, we can design a simple unbiased estimator as follows.
First compute all the degrees in $S$, and let the sum be $D$. Pick $s \in S$ with probability
$\deg(s)/D$, pick a uar edge $e$ incident to $s$, and output $D/|e|$. This is easily seen to be an unbiased estimate of $h(S)$. 
The variance calculation yields the connection to the average \emph{arity}, $\eavg$; plugging this in gives 
a clean $\Ot_\eps(\sqrt{n} + \eavg)$ algorithm
to estimate $m$. In particular, if $\eavg \leq \sqrt{n}$, then we are done. 
One is left with the case when $\eavg > \sqrt{n}$. Since the maximum arity of a hyperedge is at most $n$, 
a simple second-moment argument shows that $\eavg$ can be estimated by taking the average arity of $O_\eps(n/\eavg) = O_\eps(\sqrt{n})$ many uar 
hyperedges.
Using the handshake lemma $n \davg = m \eavg$, it now suffices to estimate $\davg$.

We next show that $\davg$ can be estimated, for any value of $\eavg$, using another simple
algorithm along the same ideas. 
We start with the obvious algorithm of sampling $\Ot_\eps(\sqrt{n})$
uar vertices and querying their degrees. The $O(\sqrt{n})$ vertices whose
degree is larger than $\sqrt{n} \davg$ create large variance, so our aim is to simply
discover all of them. Since each such vertex participates in at least $\sqrt{n} \cdot m \eavg/n = m \eavg/\sqrt{n}$ hyperedges, it suffices to sample $\Theta(\sqrt{n}/\eavg)$ uar
hyperedges and query \emph{all} their vertices. The total number of queries, in expectation, is $O(\eavg \cdot \sqrt{n}/\eavg) = O(\sqrt{n})$, 
independent of $\mu$, as desired. We can make the query complexity deterministic without performing a parameter search for $\mu$, by keeping a budget of $\Ot_\eps(\sqrt{n})$ and querying until we run out.

Our final algorithm is quite transparent, and avoids niggling complications
like parameter search, bucketing, and choosing complex thresholds. The analysis is clean and streamlined, leading to lower dependencies on $\log n ,\eps^{-1}$ factors.

\subsection{Discussion and Further Results} \label{sec:pers}

We make a few remarks and offer additional perspectives.

{\bf Knowledge of number of vertices.} The above results assume that $n$ is known. However, it is well known (see, e.g,~\cite{RoTs16}) that given access to random vertices, one can obtain a $(1+\eps)$-estimate of $n$ 
using a standard collision counting method. All our algorithms can work with an estimate of $n$, and therefore, we obtain $\Ot_\eps(\sqrt{n})$-query estimation algorithms without needing to know $n$. Furthermore, we have no restriction on isolated vertices since random hyperedge queries avoid such vertices.

{\bf Dependence on $\eps$.} We note that the query complexity of estimating the average degree in~\Cref{thm:main-avgdeg} is better than that in~\Cref{thm:main}
since $\eps^{-1}\sqrt{n\log n} \leq \eps^{-2}\sqrt{n} + \sqrt{n}\log n$ by the AM-GM inequality.
We leave open the question whether a linear dependence on $\eps$ is possible for estimating the number of hyperedges.

{\bf Estimating average arity.} As noted earlier, we have $n \davg = m \eavg$. Among the unknowns $m, \davg, \eavg$, it suffices to estimate
two of them to get an estimate of all three. So our main results prove that all of them can be accurately estimated using $\Ot_\eps(\sqrt{n})$
queries. While our theorems do not reference $\mu$, our lower bounds constructions prove that 
$\Omega(\sqrt{n})$ queries are required to estimate any one of the quantities to any constant factor, establishing the optimality of our
algorithms (up to $\log n, \eps^{-1}$ dependencies).

{\bf Non-uniformity is provably harder.} As noted earlier, estimating $m$ and $\davg$ are equivalent in $k$-uniform hypergraphs.
Furthermore, in this uniform case, the fractional degree is a scaled version of the degree, and therefore can be obtained with one query exactly.
This makes the problem much easier; indeed, the works of 
Motwani, Panigrahy, and Xu~\cite{MPX07} and that of Beretta and T\v{e}tek~\cite{BT24}
imply $O_\eps(n^{1/3})$-query algorithms for estimating the number of edges/average degree in any 
$k$-uniform\footnote{Indeed, it is not too hard to show by rejection sampling that if all hyperedges are of size $\leq k$, then one can get an $O(kn^{1/3})$-query algorithm to estimate the average degree. Without uniformity, however, this doesn't give anything for $m$-estimation.} hypergraph. Note that $k$ can be arbitrary and need not be a constant. 

The bounds of \Thm{main} and \Thm{main-avgdeg} prove that, for arbitrary hypergraphs, estimating $m$ or $\davg$ has a higher query complexity in terms of $n$ than in uniform hypergraphs. Nevertheless, our ideas can be melded with the ones in~\cite{MPX07,BT24} to prove that if the {\em average} hyperedge arity is bounded (a weaker condition than maximum edge size being bounded, and may hold true in real-world hypergraphs), then one can get an $n^{1/3}$ dependence for estimating the number of hyperedges. The dependence on $\eps$ is not great, and we leave improving this as future work (see~\Cref{sec:open}). Interestingly, estimating the {\em average degree} 
remains ``$\sqrt{n}$-hard'' since the lower bound examples establishing~\Cref{thm:main-avgdeg} have constant average hyperedge arity.

\begin{mdframed}[backgroundcolor=gray!20,topline=false,bottomline=false,leftline=false,rightline=false]
	\smallskip
	\begin{result} \label{thm:main-bndmu}
		Consider the same setup as \Cref{thm:main}. Given a constant factor upper bound on the average hyperedge arity $\eavg$, there is a randomized algorithm 
		that outputs a $(1+\eps)$-approximation to $m$ with high probability, and makes $\Ot_\eps(\eavg n^{1/3})$-many queries
		 (\Cref{thm:est-m-bnd-mu} in~\Cref{sec:est-m-bnd-mu}).
	\end{result}
\end{mdframed}

{\bf Vertex size estimation in bipartite graphs.} 
The problem
of estimating $m$ can be cast as estimating the size of a \emph{vertex} set
in the incidence graph. Consider a bipartite graph $G = (L\cup R,F)$, with vertex sets $L$ and $R$.
The dual access model is equivalent to the standard model on this graph, with one twist.
We can selectively sample uar vertices in $L$ or $R$. This provides significant power.
Our results\footnote{Our hypergraph algorithms allow for parallel copies of hyperedges,
which translates to arbitrary bipartite incidence matrices.} prove that one can estimate the size of (say) $R$ in just $\Ot_\eps(\sqrt{|L|})$ queries,
\emph{regardless} of the size of $R$! Note that in an arbitrary graph $G = (V,E)$, estimating the number of vertices in the standard access model needs $\Omega(\sqrt{|V|})$ queries (see~\cite{BCS26}).

\subsection{Related works} \label{sec:related}

The problem of estimating $m$ (or technically, the average degree) by sampling
was first initiated by Feige~\cite{Fe06} who gave a $(2+\eps)$-approximation using $\Ot_\eps(\sqrt{n/\davg})$ queries; his paper didn't allow neighbor queries. Goldreich and Ron~\cite{GR08} gave the first $\Ot_\eps(\sqrt{n/\davg})$ query 
$(1+\eps)$-approximation for $m$,
exploiting the power of neighbor queries. Eden, Ron, and Seshadhri~\cite{ERS19} gave a significantly simpler
algorithm for this problem making $O(\eps^{-2}\sqrt{n/\davg})$ queries; refer to the appendix of~\cite{BCS26} and~\cite{EdSe26} for an exposition of this algorithm. 
The problem of sampling edges uniformly at random has also received
much attention~\cite{EdRo18,EdRoRo19,EdNaTe22}.
The graph degeneracy has been shown to play a fundamental
role in approximating $m$~\cite{ERS19,EdRoRo19,TT22,Ch26}.

Motwani, Panigrahy, and Xu~\cite{MPX07} showed that \emph{proportional sampling} leads to $\Ot_\eps(n^{1/3})$
query algorithms for estimating the sum of an array. These results directly translate
to estimating $m$, by considering the array of degrees. Beretta and T\v{e}tek~\cite{BT24} give improved
algorithms, and a slick ``Harmonic estimator"
for this problem. We note that these algorithms work directly for $k$-uniform hypergraphs,
with dual access.

T\v{e}tek and Thorup~\cite{TT22} study the $m$ estimation problem in graphs under alternate models of access
for the adjacency list, focusing on the $\eps$ dependence. Beretta, Chakrabarty, and Seshadhri~\cite{BCS26}
showed that pair queries lead to $\Ot_\eps(n^{1/4})$ query algorithms and give a nuanced complexity landscape
for $m$ estimation, depending on the specific query model.
Beame, Har-Peled, Ramamoorthy, Rashtchian, and Sinha~\cite{BeHa+20} initiated a beautiful line of work on estimating $m$ using Independent Set (IS) queries,
that simply output whether a set of vertices contains an edge. There has been further work on this problem, and generalizations to BIS (Bipartite Independent Set) queries~\cite{ChLeWa20,BiGhKo+18,BhBiGh21,BhBiGh+22}.
Recently, Adar, Hotam, and Levi~\cite{AdHoLe26} gave improved algorithms when IS queries are available with the 
standard model.

As noted earlier, the question of estimating the average degree of hypergraphs in the standard access model was first discussed
by Barhum~\cite{Ba07}, for bounded arity hypergraphs. 
Dell and Lapinskas posed the question of estimating $m$
in $k$-uniform hypergraphs, using IS and Colorful IS queries~\cite{DeLa18}. Further followup has lead to near optimal algorithms~\cite{DeLaMe20,DeLaMe24}.

Directly relevant to our dual access model is the study of sublinear algorithms for the set cover problem. A set system is precisely a hypergraph, 
and sublinear algorithms for set cover have been studied in (variants of) the dual access model. This work was initiated by Nguyen and Onak~\cite{NO08} and considered later by
~\cite{YYI12,InMaRu+18}. While~\cite{NO08, YYI12} focused on bounded-degree and bounded-arity hypergraphs, the paper by Indyk, Mahabadi, Rubinfeld, Vakilian, and Yodpinyanee~\cite{InMaRu+18}
considered general hypergraphs, but their algorithms' query complexities were  sublinear in $mn$ (rather than $n$ or $m$).
Sublinear hitting set algorithms were studied by Bishnu, Ghosh, Kolay, Mishra, and Saurabh, in versions of the IS model~\cite{BiGhKo+23}.

\section{Algorithms} \label{sec:algs}

The input is a hypergraph $H = (V(H), E(H))$, where we think of vertices
and hyperedges being represented by IDs. Therefore, when we say ``a hyperedge $e$",
we mean that $e$ is the ID, not the explicit representation as a subset of $V(H)$.
This is important for sublinear algorithms, since individual hyperedges
can have unbounded size.

The following queries constitute the standard ``adjacency list" model.
\begin{itemize}[noitemsep]
	\item $\RV()$: returns a random vertex $x\in V$ uniformly at random.
	\item $\DEG(x)$: returns the number of hyperedges incident on $x$.
	\item $\IncE(x,i)$: if $i\leq \DEG(x)$, then returns the $i$th hyperedge containing $x$. (We  pick $i$ at random to get a uar hyperedge containing $x$.)
\end{itemize}

The ``dual access" provides the following additional queries.
\begin{itemize}[noitemsep]
	\item $\RE()$: returns a random hyperedge $e\in E$ uniformly at random.
	\item $\SIZE(e)$: returns the number of vertices in $e$.
	\item $\IncV(e,j)$: if $j\leq \SIZE(e)$, then returns the $j$th vertex in $e$.
\end{itemize}

In all our algorithms, we assume dual access to the $H$. We also assume that $n$
is known to the algorithms. For convenience, we do not explicitly state these
as inputs. We use ``whp" as a shorthand for probability at least $2/3$. 
We use the shorthand ``$(1\pm \eps)$-estimate for $A$" to mean
an estimate that lies in $[(1-\eps)A, (1+\eps)A]$.

We state the classic Chebyshev's inequality.
\smallskip

\begin{theorem} \label{thm:cheb} Let $X$ be a random variable.
Then for any $t>0$, $\Pr[|X - \EX[X]| \geq t] \leq \Var[X]/t^2$.
\end{theorem}

In what follows, for clarity's sake, we have not optimized constant factors. In particular, $C$, $C_0$, and $C_1$ below
are sufficiently large constants.

\subsection{Estimating Average Degree} \label{sec:est-avg}

We begin with a simple algorithm for estimating $\davg$ of an arbitrary hypergraph. We set a query budget $B = O(\sqrt{n\log n}/\eps)$, and
sample uar hyperedges and all incident vertices 
until the budget is exhausted. We query the degrees of the vertices and sum those exactly.
With high probability, this process discovers all the ``high-degree'' vertices, and computes their contribution to $n\davg$.
The remaining vertices have low degree, so the
standard empirical mean on $B$ uar vertices is a good estimate
of the remaining contribution. \smallskip

\begin{algorithm}
	\caption{\estavg$(\eps)$} \label{alg:estavg}
	\begin{algorithmic}[1]
		\State $U \eq \emptyset$, $B \eq \ceil{\frac{C\sqrt{n\log n}}{\eps}}$, and $j\eq 1$.
		\While{True}:
			\State Use $\RE$ to sample random hyperedge $e_j$ and query $|e_j| \eq \SIZE(e_j)$.
			\State If $\sum_{i=1}^j |e_i| > B$: {\bf break}. \label{alg:s4}
			\State Add all vertices in $e_j$ to $U$, making $|e_j|$ many $\IncV$ queries.
		\EndWhile
		\State Query $\DEG(x)$ for all $x\in U$; Set $\est_1 \eq \frac{1}{n}\sum_{x\in U} \deg(x)$
		\For{$i = 1, \ldots, B$}: \label{alg:s7}
			\State Use $\RV$ to sample random vertex $x\in V$.
 			\State If $x\notin U$, query $\DEG(x)$ and set $Y_i \eq \deg(x)$; otherwise, $Y_i \eq 0$.
		\EndFor
		\State $\est_2 \eq \frac{1}{B} \sum_{i=1}^B Y_i$
		\State Return $\est \eq \est_1 + \est_2$
	\end{algorithmic}
\end{algorithm}

\begin{theorem} \label{thm:estavg}
\estavg$(\eps)$ returns an
$(1\pm\eps)$-estimate of $\davg$ with probability $\geq 0.99$ making $O(\eps^{-1}\sqrt{n\log n})$ queries.  
\end{theorem}

\begin{proof}
We begin with the error analysis.
Partition the vertices into ``heavy'' and ``light'' as follows.
\[ H := \{x\in V~:~ \deg(x) > \davg \eps \sqrt{n\log n} \}~~\text{and}~~ L := V\setminus H\]
Define the good event which states that all heavy vertices are captured in the set $U$ after Step 1. That is, 
$\calE := \{H \subseteq U\}$. We assert:
\begin{lemma}\label{lem:1}
	$\Pr[\calE] \geq 0.999\cdot \left(1-1/n\right)$.
\end{lemma}
\begin{proof}
Let $q := \left\lceil C_1\frac{\sqrt {n\log n}}{\eps \eavg}	\right\rceil$ and let $e_1,\ldots,e_q$ denote the first $q$ sampled edges of Stage I. Note that this quantity is only used for analysis; the algorithm doesn't know $q$ since it doesn't know $\eavg$.
Consider the event 
$\calE_1 := \left\{\sum_{i=1}^q|e_i|\leq B\right\}$.
Note that $\calE_1$ implies that we sample at least $q$ edges in Step 1.

We show that $\Pr[\calE_1] \geq 0.999$. Since $\Exp[|e_i|] = \eavg$, we get that $\Exp[\sum_{i=1}^q |e_i|] = q\eavg \leq (C_1 + 1)\frac{\sqrt{n\log n}}{\eps}$.
If $C$ is suitably larger than $C_1$, then $\Pr[\calE_1] \geq 0.999$.
	
Now fix a vertex $x\in H$. The probability the $i$th sample $e_i$ contains $x$ is precisely $\frac{\deg(x)}{m}$.
Since $\calE_1$ implies $|U|\geq q$, the probability $x\notin U$ is 
	\[
		\Pr[x\notin U~|~\calE_1] \leq \left(1 - \frac{\deg(x)}{m}\right)^q \leq \left(1 - \frac{\davg\eps\sqrt{n\log n}}{m}\right)^q < \exp\left( -  \frac{q\mu \eps \sqrt{\log n}}{\sqrt{n}}\right) < \frac{1}{n^{2}}
	\]
	where the second inequality follows 
	using $n\davg = m\eavg$ and $(1-z)< e^{-z}$, and the third inequality follows because conditioned on $\calE_1$, $q$ is large enough. It also assumes the constant $C_1$ is large enough.
	Taking a union bound on all $x\in H$, we get $\Pr[\calE~|~\calE_1] \geq \left(1 - \frac{1}{n}\right)$, and so the lemma follows from the fact that $\Pr[\calE_1] \geq 0.999$.
\end{proof}

\noindent
Let us fix the random set $U$ after Step 1 and assume $\calE$ is satisfied; we will account for $\Pr[\neg \calE]$ in the final error calculation.
Note that for any $1\leq i\leq B$, we have 
$\Exp[Y_i] = \frac{1}{n}\sum_{x\notin U} \deg(x)$ and $\Var[Y_i] \leq \Exp[Y^2_i] = \frac{1}{n}\sum_{x\notin U} \deg^2(x)$.
Since $\est_2$ is an average of $B$ iid samples, and since $\est = \est_1 + \est_2$, we get 
\[
\Exp[\est] = \Exp[\est_1 + \est_2] = \frac{1}{n}\sum_{x\in U} \deg(x) + \frac{1}{n} \sum_{x\notin U} \deg(x) ~~=~\davg
\]
establishing that $\est$ is an unbiased estimate of the average degree. Furthermore, if $\calE$ is satisfied, any $x\notin U$ lies in $L$.
So, $\deg(x) \leq \davg\eps\sqrt{n\log n}$, for all $x\notin U$. This upper bounds the variance of $Y_i$, and thus $\est_2$:
\[
\Var[Y_i] \leq \davg\eps\sqrt{n\log n} \cdot \Exp[Y_i] \leq \davg^2\eps\sqrt{n\log n} ~~\Rightarrow~~\Var[\est_2] \leq \frac{\davg^2\eps\sqrt{n\log n}}{B} \leq 0.001 \eps^2\davg^2
\]
Since $\Var[\est_1] = 0$, once we fix $U$, we get $\Var[\est] \leq 0.001\eps^2\davg^2$. By the Chebyshev inequality, we get
%
\[
\Pr[|\est - \davg| \geq \eps \davg] \leq 0.001 
\]
Therefore, the total error probability is $\Pr[\neg \calE] + 0.001 < 0.01$ from \Lem{1} (assuming $n$ is large enough).
This completes the proof.

\noindent
{\bf Query Complexity.} Since all $|e|\geq 1$, the number of $\RE$ queries, $\IncV$ queries, and $\SIZE(e)$ queries in Step 1 is $O(\sqrt{n\log n}/\eps)$.	
The size $|U|$ is also bounded by $B$, and therefore, the number of $\DEG$ queries in Step 2 is $O(\sqrt{n\log n}/\eps)$. Finally, in Step 3(a) and (b), 
we make $O(\sqrt{n\log n}/\eps)$ many $\RV$ and $\DEG$ queries. Note that we don't make any $\IncE$ queries. In sum, the total query complexity is $O(\sqrt{n\log n}/\eps)$.
\end{proof}

\subsection{Fractional Degree and Hyperedges} \label{sec:harm}

Define the {\em fractional degree} $h(x)$ of a vertex $x$ as 
$
	h(x) := \sum_{e \ni x} \frac{1}{|e|}.
$
The simple but {\em key} observation is:
\begin{equation}\label{eq:h-to-m}
	\sum_{x\in V} h(x) = m
\end{equation}
This converts the problem to estimating the sum of a vector $h$ indexed by vertices of $V$. 
For any subset $S \subseteq V$, let $h(S) := \sum_{s \in S} h(s)$.
We cannot evaluate $h(s)$ in $O(1)$ queries, but can design a simple unbiased estimator
for $h(s)$. Our main tool is the following algorithm, that estimates $h(S)$ for any subset $S$.
It is convenient to have an input parameter $b$ that decides the sampling budget.

We use $\deg(S)$ to denote $\sum_{s \in S} \deg(s)$.

\begin{algorithm}
	\caption{\estH$(S,b)$} \label{alg:estH}
	\begin{algorithmic}[1]
        \State Query $\DEG(s)$ for all $s\in S$ and compute $\deg(S) \eq \sum_{s\in S} \deg(s)$.
		\For{$i = 1, \ldots, b$}:
			\State Sample a vertex $s\in S$ with probability $\frac{\deg(s)}{\deg(S)}$.
 			\State Make an $\IncE(s)$ query to sample uar hyperedge $e \ni s$.
            \State Set $Y_j \eq \frac{\deg(S)}{|e|}$.
		\EndFor
		\State Return $Y \eq \frac{1}{b}\sum_{j=1}^b Y_j$
	\end{algorithmic}
\end{algorithm}

\begin{lemma} \label{lem:estH} The procedure \estH$(S,b)$ makes $|S|+b$ queries.
Let $Y$ denote the output of \estH$(S,b)$. Then, for any $t>0$, we have $\Pr[|Y - h(S)| \geq t] \leq \frac{\deg(S) h(S)}{bt^2}$.
\end{lemma}

\begin{proof} The query bound is easy to see; there are $|S|$ degree queries
made and exactly $b$ queries to $\IncE$.
For a fixed $j$, let us calculate the mean and variance of $Y_j$. Note that
	\[
	\Exp[Y_j] ~=~ \sum_{s\in S} \frac{\deg(s)}{\deg(S)} \cdot \sum_{e\ni s} \frac{1}{\deg(s)}\cdot \frac{\deg(S)}{|e|} ~=~ \sum_{s\in S} \underbrace{\sum_{e\ni s} \frac{1}{|e|}}_{=h(s)}~=h(S)
	\]
Next, note that
	\[
	\Exp[Y_j^2] ~=~ \sum_{s\in S} \frac{\deg(s)}{\deg(S)} \cdot \sum_{e\ni x} \frac{1}{\deg(s)}\cdot \frac{\deg(S)^2}{|e|^2} ~\leq ~ \deg(S) \cdot \sum_{s\in S}\underbrace{\sum_{e\ni s} \frac{1}{|e|}}_{=h(s)} ~=~\deg(S) h(S)
	\] 
	where we used $|e|\geq 1$ in the inequality.
Therefore, $\Exp[Y] = h(S)$ and $\Var[Y] \leq \frac{\deg(S) h(S)}{b}$. A direct
application of Chebyshev's inequality completes the proof.
%
%
\end{proof}

\subsection{Estimating Number of Hyperedges} \label{sec:wrapup}

We wrap everything up. We begin with a standard observation that if $\eavg = \Omega(\sqrt{n})$, then averaging the size of $O(\eps^{-2}\sqrt{n})$ random hyperedges is a $(1\pm \eps)$-estimate.
The following subroutine either provides a good estimate to $\eavg$ or asserts that $\eavg \leq \sqrt{n}$.

\begin{algorithm}
	\caption{\estmu$(\eps)$} \label{alg:estmu}
	\begin{algorithmic}[1]
 		\State Sample $r := \ceil{\eps^{-2} C_0 \sqrt{n}}$ uar edges $e_1, \ldots, e_r$ using $\RE()$ queries.
		\State Obtain all $|e_i|$ values using $\SIZE()$ queries.
		\State Evaluate $X \eq \frac{1}{r} \sum_{i=1}^r |e_i|$
		\State If $X > 0.5\sqrt{n}$, return $X$ asserting it is a $(1\pm \eps)$-estimate.
		\State Else: assert $\eavg \leq \sqrt{n}$.
	\end{algorithmic}
\end{algorithm}

\smallskip

\begin{lemma}\label{lem:estmu} The procedure \estmu$(\eps)$ outputs a correct assertion
with probability $\geq 0.99$, and makes $O(\eps^{-2}\sqrt{n})$ queries.
\end{lemma}
\begin{proof} The query complexity bound is immediate, since exactly $r$ $\RE()$ and $\SIZE()$ queries are made.

We have $\Exp[X] = m^{-1}\sum_{e\in E} |e| = \eavg$.
Let $X_i$ be the random variable $|e_i|$, so $\Var[X_i] \leq \Exp[X^2_i]
\leq \max X_i \cdot \Exp[X_i] \leq n \eavg$. Hence, $\Var[X] \leq n \eavg/r$.
By Chebyshev's inequality,
$\Pr[|X - \eavg| \geq \eps \eavg] \leq \frac{n}{r\eps^2 \eavg}$.
If $\eavg > 0.005\sqrt{n}$, then for $r = \ceil{\eps^{-2} C_0 \sqrt{n}}$ (and sufficiently large $C_0$)
$r \eps^2 \eavg > 100n$. So, $X$ is a $(1\pm \eps)$-estimate with at least $0.99$ probability.

If $\eavg < 0.005\sqrt{n}$, then $\Pr[X > 0.5\sqrt{n}] \leq 0.01$ by the Markov
inequality. With at least $0.99$ probability, the algorithm correctly asserts that $\eavg \leq \sqrt{n}$.
If $\eavg \in [0.005\sqrt{n}, \sqrt{n}]$: if $X$ is output, the assertion
is correct with probability at least $0.99$. Otherwise, the correct assertion ``$\eavg \leq \sqrt{n}$" is output.
If $\eavg > \sqrt{n}$, then with probability at least $0.99$, $X$ is a $(1\pm\eps)$-approximation
for $\eavg$ and is hence at least $0.5\sqrt{n}$. So the output is correct.
\end{proof}

The full algorithm to estimate the number of hyperedges is the following.

\begin{algorithm}
	\caption{\estHyp$(\eps)$} \label{alg:esthypedge}
	\begin{algorithmic}[1]
		\State Call \estavg$(\eps)$ to get estimate $\widehat{\davg}$.
		\State Call \estmu$(\eps)$. 
		\If{\estmu\ outputs estimate $\widehat{\mu}$}
		\State \Return $n\widehat{\davg}/\widehat{\mu}$.
		\EndIf
		\LineComment{\estmu\ asserts that $\eavg \leq \sqrt{n}$.}
		\State $q \eq \ceil{C_1 \sqrt{n} \log n}$, $U \eq \emptyset$, $b \eq \ceil{C_1 \eps^{-2} \sqrt{n}}$, $S \eq \emptyset$
		\For{$i = 1,\ldots, q$}: 
			\State Sample hyperedge $e_i$ using $\RE()$.
			\State Get uar $u \in e_i$ using $\IncV(e_i)$ query, and add $u$ to $U$.
		\EndFor
		\State Get $Y_U$ as the output of \estH$(U,b)$.
		\For{$i = 1,\ldots,b$}:
			\State Get uar vertex $s_i$ using $\RV()$. If $s_i \notin U$, add $s_i$ to $S$.
		\EndFor
		\State Get $Y_S$ as the output of \estH$(S,b)$.
		\State \Return $(n/b)Y_S + Y_U$.
	\end{algorithmic}
\end{algorithm}

\begin{theorem} \label{thm:esthyp} The output is \estHyp$(\eps)$
is a $(1\pm2\eps)$-estimate for $m$ with probability at least
$2/3$. The query complexity is $O(\eps^{-2} \sqrt{n} + \sqrt{n}\log n)$.
\end{theorem}

\begin{proof} 
The query complexity follows from an inspection of the algorithm, and from the query complexity of \estavg($\eps$) and \estmu($\eps$).
Also recall that $\eps{-1}\sqrt{n\log n} \leq \eps^{-2}\sqrt{n} + \sqrt{n}\log n$ by AM-GM inequality. We now move to the error analysis.

Consider the two function calls, \estavg$(\eps)$ and \estmu$(\eps)$.  By \Thm{estavg} and \Lem{estmu},
and the union bound, the probability of both procedures being
correct is at least $0.98$. Let us condition on this event.

If \estmu$(\eps)$ outputs an estimate $\widehat{\mu}$,
then by \Lem{estmu}, it is a $(1\pm\eps)$-estimate for $\mu$.
By \Thm{estavg}, $\widehat{\davg}$ is a $(1\pm\eps)$-estimate
for $\davg$. Since $m = n\davg/\mu$, the output $n\widehat{\davg}/\widehat{\mu}$ is a $(1\pm 2\eps)$-estimate for $m$.

Now, suppose that \estmu$(\eps)$ asserts that $\mu \leq \sqrt{n}$ (we conditioned on the assertion being correct). We focus on $Y_S$ and $Y_U$
separately, and break up the analysis into a series of claims.

\begin{claim} \label{clm:YU} $\Pr[|Y_U - h(U)| \geq \eps m/2] \leq 0.01$.
\end{claim}

\begin{proof}  
	This follows by applying the trivial upper bounds of $\deg(U) \leq n\davg$ and $h(U) \leq m$ to the RHS of 
	\Lem{estH}, to give
$$ \Pr[|Y_U - h(U)| \geq \frac{\eps m}{2}] \leq \frac{4 n \davg m}{b\eps^2 m^2} = \frac{4\mu}{b\eps^2} $$
The choice of $b = \ceil{C_1 \eps^{-2} \sqrt{n}}$ and
the fact $\sqrt{n} \geq \mu$ ensures this probability
is at most $0.01$.
\end{proof}

We now state the crucial property that $U$ contains
all vertices $x$ that have large $h(x)$ value.

\begin{claim} \label{clm:U} With probability at least $1-1/n$,
for all vertices $x \notin U$, $h(x) \leq m/\sqrt{n}$.
\end{claim}

\begin{proof}
Consider a vertex $x$.
It is added to $U$ in the $i$th iteration, if the 
uar edge $e_i$ contains $x$, and $x$ is sampled through
the $\IncV(e_i)$ query.
The probability of this event is exactly $m^{-1} \sum_{e \ni x} 1/|e| = h(x)/m$. 
For any $x$ with $h(x) > m/\sqrt{n}$, the probability 
that $x$ is sampled in an iteration is $\geq 1/\sqrt{n}$. The probability that none of the $q$ independent samples places it in $U$ is $\leq (1 - 1/\sqrt{n})^q \leq 1/n^2$ for a suitably large constant $C_1$. Taking a union bound over all $x$,
with probability at least $1/n$, all $x$ such
that $h(x) > m/\sqrt{n}$ are present in $U$.
\end{proof}

For convenience, define $\tilde{h}(x)$ to be $h(x)$
if $x \notin U$ and zero otherwise. 
Observe that $m = h(U) + \tilde{h}(V)$.
By \Clm{U}, $\max_x \tilde{h}(x) \leq m/\sqrt{n}$.
This bound plays a critical role in the following bound
on $h(S)$.

\begin{claim} \label{clm:S} $\Pr[|(n/b) h(S) - \tilde{h}(V)| \geq \eps m/4] \leq 0.01$.
\end{claim}

\begin{proof}
It is convenient to define $\tilde{S}$ as the sample
set $\{s_1, s_2, \ldots, s_b\}$. Note that $S = \tilde{S} \setminus U$ and $\tilde{h}(S) = h(S)$.
Observe that $\Exp[\tilde{h}(s_i)] = \tilde{h}(V)/n$
and $\Var[\tilde{h}(s_i)] \leq \max \tilde{h}(s_i) \Exp[\tilde{h}(s_i)]$ $\leq (m/\sqrt{n}) \tilde{h}(V)/n$. 
The bound on $\tilde{h}(x)$ bounds this variance.

Thus, $\Exp[\tilde{h}(S)] = \tilde{h}(V)b/n$
and $\Var[\tilde{h}(S)] \leq (m/\sqrt{n}) \tilde{h}(V) b/n$.
Applying Chebyshev's inequality and noting that $\tilde{h}(V) \leq m$,
\begin{equation*}
\Pr[|\tilde{h}(S) - \tilde{h}(V)b/n| \geq \eps mb/4n]
\leq \frac{16(m/\sqrt{n}) \tilde{h}(V) b/n}{\eps^2 b^2 m^2/n^2}
\leq \frac{16m^2 \sqrt{n}}{b\eps^2 m^2} = \frac{16\sqrt{n}}{b\eps^2} \leq 0.01
\end{equation*}
Thus, $\Pr[|(n/b) h(S) - \tilde{h}(V)| \geq \eps m/4] = \Pr[|\tilde{h}(S) - \tilde{h}(V)b/n| \geq \eps mb/4n] \leq 0.01$.
\end{proof}

We have the tools to analyze $Y_S$.

\begin{claim} \label{clm:YS} $\Pr[|(n/b) Y_S - \tilde{h}(V)| \geq \eps m/2] \leq 0.05$.
\end{claim}

\begin{proof}
Note that $\EX_S[\deg(S)] \leq |S| \davg \leq b \davg$ and $\EX_S[h(S)] \leq bm/n$.
By the Markov inequality and an union bound, with probability at least $0.98$,
$\deg(S) \leq 100 b\davg$ and $h(S) \leq 100 bm/n$. Let us condition on 
such a ``good" choice of $S$.
We apply \Lem{estH} and use $C$ to denote a sufficiently large constant.
\begin{equation*}
	\Pr[|Y_S - h(S)| \geq \eps (m/4)(b/n)] \leq \frac{16 \deg(S) h(S)}{ b \cdot \eps^2 b^2 m^2/n^2}
\leq \frac{C b^2 \davg m/n}{b^3 \eps^2 m^2/n^2} = \frac{C n \davg}{b \eps^2 m} = \frac{C \eavg}{b \eps^2} \leq 0.01
\end{equation*}
By multiplying with $(n/b)$, $\Pr[|(n/b)Y_S - (n/b)h(S)| \geq \eps m/4] \leq 0.01$,
for such a ``good" $S$. The probability of choosing a good $S$ is at leat $0.98$.
Applying the bound of \Clm{S} and a union bound, $\Pr[|(n/b) Y_S - \tilde{h}(V)| \geq \eps m/2] \leq 0.05$.
\end{proof}

We put the pieces together. Recall that $m = h(U) + \tilde{h}(V)$. The output
estimate is $Y_U + (n/b) Y_S$. So $|Y_U + (n/b) Y_S - m| \leq |Y_U - h(U)| + |(n/b) Y_S - \tilde{h}(V)|$. Applying \Clm{YU}, \Clm{YS}, and a union bound, 
$\Pr[|\left(Y_U + (n/b) Y_S\right) - m| \geq \eps m] \leq 0.1$.\qedhere

\end{proof}

\section{Lower Bounds}

\subsection{Estimating Average Degree and Average Size}\label{sec:lb-est-avgd}

\begin{theorem}\label{thm:lb-est-avgd}
	For any constant $c$, any algorithm returning an estimate $\davg \leq \widehat{\davg} \leq c\cdot \davg$ with probability $> 2/3$, 
	must make $\Omega(\sqrt{n}/c)$ many queries. 
	This holds even when $\davg$ is $O(\sqrt{n})$ and even when $\eavg = O(1)$.
\end{theorem}
\begin{proof}
	For contradiction's sake, suppose there is an algorithm $\cA$ that returns a $c$-approximation to the average degree of any hypergraph $H$ with probability $> 2/3$ and makes $q := o(\sqrt{n}/c)$-many queries.
	
	For simplicity, fix a perfect square $n$.
	We describe two (coupled) distributions, $\calD_1$ and $\calD_2$, of hypergraphs on $n$ vertices $V$, with the following properties: (i) for any $H_1 \in \supp(D_1)$ and $H_2 \in \supp(D_2)$,
	we have $\davg(H_2)/\davg(H_1) > c$, but (ii) $\cA$ cannot distinguish whether the hypergraph $H$ it is making queries on has been drawn from $\cD_1$ or $\cD_2$, with probability $> 2/3$. 
	This would complete the proof.

Fix a parameter $d \leq \sqrt{n}$. We first sample a random subset $A\subseteq V$ of $2c\sqrt{n}$ vertices. 
On the vertices $V\setminus A$, we fix a standard graph $H_0$ of average degree $d$ on $n-2c\sqrt{n}$ vertices.
We also fix a standard graph $H'$ of average degree $d$ on $2c\sqrt{n}$ vertices. Finally, we also fix a collection of $k:= d\sqrt{n}$ 
hyperedges each of size $|e| = c\sqrt{n}$. Call this hypergraph $H''$.
To describe an $H_1 \sim \cD_1$, we simply output $H_0[V\setminus A] \cup H'[A]$, where $H_0[V\setminus A]$ is the graph obtained using a random mapping between the labels of $V(H_0)$ and $V\setminus A$.
Similarly, define $H'[A]$. To describe $H_2 \sim D_2$, we output $H_0[V\setminus A] \cup H''[A]$, where we couple $H_1$ and $H_2$ using the same (i) choice of $A$, and (ii) random-mapping between $V(H_0)$ and $V\setminus A$. 

By design, note that for $(H_1, H_2)$ drawn as above, we have $\davg(H_1) = d$ while
\[
\text{for any $H\in \supp(\calD_2)$}, ~\davg(H) = \frac{1}{n} \cdot \left(d(n - |A|) + (d\sqrt{n}) \cdot c\sqrt{n} \right) \approx (c+1)d
\]
We also note that $m(H_0) = \frac{nd}{2}\cdot (1-o(1))$ and $m(H'), m(H'') \leq d\sqrt{n}$.

Since $\davg(H_2)/\davg(H_1) > c$, the $c$-approximation algorithm should distinguish $H_1$ and $H_2$ with probability $\ge 2/3$. 
We now show this is impossible.

	Let $A$ be the random subset that is shared between $(H_1, H_2)$. 
	Note that unless the algorithm obtains vertex $x\in A$, we have $\DEG(x)$, $\IncE(x)$ the {\em same} whether we call on $H_1$ or $H_2$.
	Similarly, $\IncV(e)$ is the {\em same} unless $e\subseteq A$.
	Let $\calE$ be the event that none of the algorithm's $\leq q$-many $\RV()$ calls return a vertex in $A$, and none of the $\leq q$-many $\RE()$ calls return an $e\subseteq A$.
	Then, the total variational distance between the transcripts (answers to queries), $\Pi_1$ and $\Pi_2$, when the hypergraph is $H_1$ and $H_2$ is
	\[
		d_{\mathsf{TV}} (\Pi_1, \Pi_2) \leq \Pr[\neg \calE]~~~\Rightarrow~~\Pr[\cA~\text{can distinguish $H_1$ and $H_2$}] \leq \frac{1}{2}\cdot \left(1 + \Pr[\neg \calE]\right)
	\]
	
	We now show $\Pr[\calE] \ge 1 - o(1)$ which proves the claim.
	Fix any $q$ queries to $\RV$ with $q = o(\sqrt{n}/c)$.
	The probability a $\RV$-query returns a vertex in $A$ is $|A|/n = \frac{2c}{\sqrt{n}}$. The probability none of the $q$-queries give this is $(1 - \frac{2c}{\sqrt{n}})^q \geq 1 - o(1)$ if $q = o(\sqrt{n}/c)$.
	Fix any $q$ queries to $\RE$. The probability a $\RE$ query returns an edges in $H'$ or $H''$ is at most $m(H')/m(H) \leq (2 + o(1))/\sqrt{n}$. 
	A similar calculation in the previous line completes the contradiction.

Note that the above bound holds even if the algorithm was promised that the average degree was $\geq d$, for whichever parameter $d$. In particular, it rules out $O(\sqrt{n/d})$-query $O(1)$-approximation algorithms for 
estimating average degree. This is in contrast to the graph case where one can obtain $\Ot_\eps(\sqrt{n/\davg})$-query $(1\pm \eps)$-estimation even with $\RV, \DEG$ and $\IncE$ queries.

Furthermore, we also note that for every $H\in \supp(\calD_1)$, we have $\eavg(H_1) = 2$, while for every
\[
H\in \supp(\calD_2), ~\text{we have}, ~~\eavg(H) = \frac{ d(n-2c\sqrt{n}) + c\sqrt{n} \cdot d\sqrt{n}}{d(n-2c\sqrt{n})/2 + d\sqrt{n}} \approx 2c
\]

This implies two things. (1) Obtaining a constant approximation for $\eavg$ also requires $\Omega(\sqrt{n})$ queries, even when it is guaranteed to be $O(1)$
	(2) Obtaining a constant approximation for $\davg$ requires $\Omega(\sqrt{n})$ queries, even when $\eavg$ is $O(1)$. On the other hand, our result in~\Cref{sec:est-m-bnd-mu} shows that it is 
	possible to obtain an $(1\pm \eps)$-estimate to $m$, the number of hyperedges, in $\Ot_\eps(n^{1/3})$-many queries if $\eavg = O(1)$.
	These observations together complete the whole proof.
\end{proof}

\subsection{Estimating Number of Hyperedges}\label{sec:lb-est-m}

\begin{theorem}\label{thm:lb-est-m}
For any constant $c$, any algorithm returning an estimate $m \leq \widehat{m} \leq c\cdot m$ with probability $> 2/3$, 
must make $\Omega(\sqrt{n}/c)$ many queries. 
\end{theorem}

\begin{proof}
	The set up of the proof is similar to the one in the previous section. 
	For contradiction's sake, suppose there is an algorithm $\cA$ which with probability $> 2/3$ returns a $c$-approximation to the average degree of any hypergraph $H$ and makes $q := o(\sqrt{n}/c)$-many queries.
	Fix an even square number $n$.
	We describe two (coupled) distributions, $\calD_1$ and $\calD_2$, of hypergraphs on $n$ vertices $V$, with the following properties: (i) for any $H_1 \in \supp(D_1)$ and $H_2 \in \supp(D_2)$,
	we have $m(H_2)/m(H_1) > c$, but (ii) $\cA$ cannot distinguish whether the hypergraph $H$ it is making queries on has been drawn from $\cD_1$ or $\cD_2$, with probability $> 2/3$. 
	This would complete the proof.

We now describe how to obtain $(H_1, H_2)$ in the coupled distribution.
First, we sample $\sqrt{n}-c$ many {\em equipartitions} $(S_i, V\setminus S_i)$ with $|S_i| = n/2$, for $1\leq i\leq \sqrt{n}-c$, and add $S_i$ and $V\setminus S_i$
as hyperedges. Let us call these edges $E_0$ and these are common to $H_1$ and $H_2$.
Next, to $H_1$, we (i) add a randomly sampled perfect matching of simple $2$-edges, and (ii) sample another $c-1$ equipartitions $(S_j, V\setminus S_j)$ with $|S_j| = n/2$, for $1\leq j\leq c-1$, and add $S_j$ and $V\setminus S_j$ to the hyperedge set. Let these edges be $E_1$.
To obtain $H_2 \in \calD_2$, we instead sample $c$ random disjoint perfect-matchings 
and add these edges to the set of hyperedges. Let these edges be $E_2$. 

Note that the number of edges are
\[
m(H_1) = \frac{n}{2} + 2(\sqrt{n}-1) = \frac{n}{2}\cdot (1+o(1))~~~\text{and}~~~m(H_2) = \frac{cn}{2} + 2(\sqrt{n} - c) = \frac{cn}{2}\cdot (1+o(1))
\]
and that $\deg(x) = \sqrt{n}$ for all vertices in both hypergraphs.
Since the number of edges differ by a factor of $c$, the algorithm $\cA$ should be able to distinguish between $H_1$ and $H_2$ with probability $> 2/3$.
We now show this is not possible.

As in the previous proof, we consider the transcripts $(\Pi_1, \Pi_2)$ of the queries to $(H_1, H_2)$, and identify events such that conditioned on them occurring, the transcripts are the same.
Let us first consider vertex-facing queries. $\DEG(x)$ are all the same in $H_1$ and $H_2$, and this is with probability $1$.
For a vertex $x$, we let $\calE_x$ denote the event that all $\IncE(x)$-queries made by $\cA$ returns an edge in $E_0$; we let $\calE$ denote the event $\bigcup_{x\in V} \calE_x$.
Note that $\calE$ implies that the transcript of $\IncE(x)$-queries are same for $H_1$ and $H_2$. 

Since the number of edges of $E_0$ incident to $x$ is $\sqrt{n} - c$, and $\deg(x) = \sqrt{n}$, a single $\IncE(x)$ query returns an edge in $E_0$ with probability $1 - \frac{c}{\sqrt{n}}$.
Therefore, $\Pr[\calE_x] = \left(1 - \frac{c}{\sqrt{n}}\right)^{q(x)} \geq 1 - \frac{cq(x)}{\sqrt{n}}$, where $q(x)$ is the number of $\IncE(x)$ queries.
By union bound, $\Pr[\calE] \geq 1 - \sum_{x\in V} \frac{cq(x)}{\sqrt{n}} \geq 1 - \frac{cq}{\sqrt{n}}$, and this is $1-o(1)$, if $q = o(\sqrt{n}/c)$.

Let us now consider the edge-facing queries. Let us assume the algorithm makes $q$ many $\RE$ calls, and let us assume whenever you make such a call, the algorithm also gets $|e|$ and all the vertices in $e$ for free (so that we don't have to worry about $\SIZE$ and $\IncV$ queries).
Let $\calF$ be the event that (i) none of the $\RE$ calls gives a large edge (that is, not a $2$-edge), and (ii) the $2q$ vertices participating in the $q$ simple edges are all distinct, irrespective of the input being $H_1$ or $H_2$.
Conditioned on $\calF$, the transcript of the answers to the $\RE$ queries have the same distribution in $H_1$ and $H_2$: by symmetry, it is {\em uniform} over all $q$-length tuples of pairs of vertices with all $2q$ distinct vertices. Therefore, we get that
\[
d_{\mathsf{TV}} (\Pi_1, \Pi_2) \leq \Pr[\neg \calE] + \Pr[\neg \calF]\Rightarrow~~\Pr[\cA~\text{can distinguish $H_1$ and $H_2$}] \leq \frac{1}{2}\cdot \left(1 + \Pr[\neg \calE] + \Pr[\neg \calF]\right)
\]
We have already proved the first term is $o(1)$.
All that remains is to prove $\Pr[\calF] = 1-o(1)$. We do so next and this completes the proof.

First, we show that in $H_1$ we only get $2$-edges with probability $1-o(1)$. Since we have $n/2$ total $2$-edges and $2(\sqrt{n}-1)$ big edges, the probability of getting a big edge is $\frac{2(\sqrt{n}-1)}{n/2 + 2(\sqrt{n}-1)} < \frac{2\sqrt{n}}{n/2} = \frac{4}{\sqrt{n}}$. The probability of getting a small edge is $1-\frac{2(\sqrt{n}-1)}{n/2 + 2(\sqrt{n}-1)} > 1-\frac{4}{\sqrt{n}}$. Thus, the probability of getting all small edges in $q$-queries of $\RE$ is at least $(1-\frac{4}{\sqrt{n}})^q \geq 1 - o(1)$ if $q = o(\sqrt{n}/c)$.
Note that $H_2$ has more $2$-edges and fewer big edges, so the probability of all getting $2$-edges is higher and thus also $\geq 1 - o(1)$.

Next, we assume that all sampled edges are $2$-edges since this is the case with probability $1-o(1)$. Now we want to show that the $2q$ vertices in these $q$ $2$-edges are all distinct with probability $1 - o(1)$. 
This calculation is a simple (reverse direction of) ``birthday paradox'' style calculation. Let $\calF_j$ denote the event that the $j$th random edge 
does not intersect the previous $(j-1)$ edges. 

For $H_1$, we see that
\[
\Pr[\calF] = \Pr[\wedge_{j=1}^q \calF_j ] = \prod_{j=1}^q \Pr[\calF_j~|\calF_1, \ldots, \calF_{j-1}] \geq \prod_{j=1}^q \frac{n/2-j}{n/2} \geq 1 - \frac{q^2}{n/2} = 1-o(1).
\]

For $H_2$ we see that
\[
\Pr[\calF] = \Pr[\wedge_{j=1}^q \calF_j ] = \prod_{j=1}^q \Pr[\calF_j~|\calF_1, \ldots, \calF_{j-1}] \geq \prod_{j=1}^q \frac{cn-2cj}{cn} \geq 1 - \frac{2q^2}{n} = 1-o(1)
\]
where the first inequality follows since each $2$-edge can intersect at most $2c$ edges from other matchings, and thus $j$ different $2$-edges can have at most $2cj$ many $2$-edges intersecting them.
\end{proof}

\section{Estimating $m$ with promised upper bound on $\eavg$}\label{sec:est-m-bnd-mu}
Suppose we are given a value $\eavg_0$ satisfying $\eavg \leq \eavg_0$. We also assume $\eps < 1/2$.
We now show an $\Ot_\eps(\eavg_0\cdot n^{1/3})$-query $(1\pm \eps)$-estimate using 
Beretta-T\v{e}tek~\cite{BT24} algorithm, which is henceforth abbreviated as BT. We recall their result: given any non-negative vector $w \in \RR^n_{\geq 0}$
such that one can sample a uniform entry of $w$ and an entry of $w$ with probability proportional to its entry, then there is a randomized algorithm that outputs an $(1\pm \eps)$-estimate
to $W := \sum_{x\in n} w(x)$ with high probability making $O(n^{1/3}\eps^{-4/3})$ many queries.

\smallskip

\begin{theorem}\label{thm:est-m-bnd-mu}
	Given $\eavg_0 \geq \eavg$, there is a randomized algorithm which returns an $(1\pm \eps)$-estimate to $m$ with high probability making 
	$O(\eps^{-13/3}\eavg_0\cdot n^{1/3}\log n)$-many dual access oracle queries.
\end{theorem}
\begin{proof}
As noted in the proof of \Clm{U}, we can sample an $x\in V$ proportional to $h(x)$ using $O(1)$ queries. 
Therefore, if we could obtain $h(x)$, or even an $(1\pm \eps)$-estimate of it, for any $x$ in $O(q)$-queries, then we could simply apply BT to get an $O(qn^{1/3})$-query algorithm. 

Call a vertex $x\in V$ {\em bad} if $h(x) \leq \frac{\eps}{\eavg_0}\cdot \deg(x)$. Let $B$ be the set of bad vertices, and let $G := V\setminus B$.
An easy calculation shows
\begin{equation}\label{eq:bad}
	\sum_{x\in B} h(x) \leq \frac{\eps}{\eavg_0} \sum_{x\in B} \deg(x) \leq \frac{\eps}{\eavg_0} \cdot n\davg = \frac{\eps}{\eavg_0} \cdot m\eavg \leq \eps m
\end{equation}
This implies that we can ignore the contribution of bad vertices and only incur an $\eps m$ additive error. On the other hand, if $x\in G$, then we can estimate $h(x)$ up to $(1\pm \eps)$-factor 
in roughly $\frac{\eavg_0}{\eps^3}$-queries, as we show below. This can then be bootstrapped with BT.\smallskip

\begin{lemma}\label{lem:srt}
	There is a subroutine which
	given any $x\in V$, in $O(\frac{\eavg_0 \log n}{\eps^3})$-many queries, either returns (a) $x$ is bad, or (b) returns an $\widehat{h}(x)$ which is an $(1\pm \eps)$-estimate of $h(x)$. The error probability is $\leq 1/\poly(n)$.
\end{lemma}
\begin{proof}
	The algorithm runs $L := C\ceil{\log n}$ batches. In each batch $\ell$, it samples $q := \ceil{\frac{4C\eavg_0}{\eps^3}}$ many random edges $e_1, \ldots, e_q$ incident on $x$ and let 
	\[
	h_\ell(x) := \frac{1}{q} \sum_{i=1}^q \frac{\deg(x)}{|e_i|}
	\]
	Note that $\Exp[h_\ell(x)] = h(x)$, since $\deg(x)/|e|$ for a random $e$ incident on $x$ is an unbiased estimate of $h(x)$.
	If $\geq L/3$ different $h_\ell(x)$ happen to be $< \frac{\eps}{2\eavg_0}\cdot \deg(x)$, we assert $x$ is bad. Otherwise, we return
	\[
	\widehat{h}(x) = \mathsf{median}\left(h_\ell(x)~:~1\leq \ell \leq L\right)
	\]
	We now argue correctness for the three possible cases of $h(x)$.
	
	{\bf Case 1:} $h(x) \leq \frac{\eps}{4\eavg_0} \cdot \deg(x)$. In this case, $\Pr[h_\ell(x) \geq \frac{\eps}{2\eavg_0}\cdot \deg(x)] \leq \frac{1}{2}$, 
	and so we expect $\geq L/2$ of the $h_\ell(x)$'s to be $< \frac{\eps}{2\eavg_0}\cdot \deg(x)$. By Chernoff bounds, the probability we see $< L/3$ of them 
	is $\leq  1/\poly(n)$ for large enough $C$. So in this case, our algorithm returns BAD with all but $1/n^2$ probability. In case it doesn't return BAD, we don't make any guarantees, and charge it to the error probability.
	
	{\bf Case 2:} $\frac{\eps}{4\eavg_0} \cdot \deg(x) < h(x) \le \frac{\eps}{\eavg_0} \cdot \deg(x)$. In this case, our algorithm could assert $x$ is BAD, in which case it would be correct. If it does return an estimate, then we know the variance of $\deg(x)/|e|$ is $\leq \deg(x) h(x)$, so for
	each $h_\ell(x)$ batch, the variance is 
	\[
	\leq \frac{\deg(x) h(x)}{q} 
	\le
	\frac{\eps \cdot \deg(x)}{4 \eavg_0} \cdot \frac{h(x) \cdot \eps^2}{C}
	<
	\frac{h^2(x) \cdot \eps^2}{C}
	\]
	by the given that $\frac{\eps}{4\eavg_0} \cdot \deg(x) < h(x)$. Thus, $h_\ell(x)$ is a $(1\pm \eps)$-estimate to $h(x)$ with $\geq 5/6$ probability, and so the median of $C\log n$ iid trials is a $(1\pm \eps)$-estimate
	with $1 - 1/\poly(n)$ probability.
	
	{\bf Case 3:} $h(x) > \frac{\eps}{\eavg_0} \cdot \deg(x)$. In this case, we should not assert BAD; the Case 2 analysis holds for the estimate $\widehat{h}(x)$ in this case. Furthermore, as in Case 2, each $h_\ell(x)$ is in $(1\pm \eps)h(x)$ with $5/6$ probability. In particular, for any fixed $\ell$, the chance that $h_\ell(x) < \frac{\eps}{2\eavg_0}\cdot \deg(x)$ is at most $1/6$.
	So we expect to see $\leq L/6$ many $h_\ell(x)$'s; the probability this $> L/3$, by Chernoff bounds is $\leq 1/\poly(n)$.
\end{proof}

Given the low error probability, let us condition on the good event and henceforth assume the above subroutine never errs, and in particular, for every good vertex it returns an accurate estimate.
Our final algorithm runs BT on the implicit vector $h'(x)$ where $h'(x) = h(x)$ when $x\in G$ but $h'(x) = 0$ when $x\in B$.
\eqref{eq:bad} and our conditioning implies that $\sum_{x\in V} h'(x) \geq (1-\eps)m$. 

Whenever the BT algorithm wishes to obtain $h'(x)$, we run the algorithm in~\Cref{lem:srt}; if we are in case (a), we return $0$, otherwise we return $\widehat{h}(x)$. Whenever we wish to sample proportional to $h'(x)$, we sample proportional to $h(x)$, and then run the algorithm in~\Cref{lem:srt}.
If we obtain case (a), then we reject and go again; by \eqref{eq:bad}, we will get a sample $x$ proportional to $h'(x)$ in $O(1)$-expected time.
Thus, we can get a $(1\pm \eps)$-estimate of $\sum_{x\in V} h'(x)$ in $O(n^{1/3}\eps^{-4/3})$ many calls, where each call makes $O(\eps^{-3}\eavg_0\log n)$-many queries.
Overall this is the desired query complexity.
\end{proof}

\section{Conclusions and Future Directions} \label{sec:open} 

In this work, we introduced the dual access query model for arbitrary hypergraphs and considered the problem of estimating the number of hyperedges and the average degree.
 We gave $\Ot_\eps(\sqrt{n})$-query algorithms for both problems and proved matching $\Omega(\sqrt{n})$ lower bounds for constant factor approximations. Our results settle the worst-case query complexity dependence on $n$ up to logarithmic factors. Nevertheless, many interesting questions remain and we conclude the paper with several directions. Unless otherwise stated, the problem is to get $(1+\eps)$-approximations to $m$ in hypergraphs.

\begin{asparaitem}
	\item {\em Going between $n^{1/3}$ and $\sqrt{n}$.} As we mention in~\Cref{sec:pers}, 
	if the hypergraph is uniform, then one can obtain $\Ot_\eps(n^{1/3})$-query algorithms 
	using previous works~\cite{MPX07, BT24} as a black box. Our \Cref{thm:main} gives a $\Ot_\eps(\sqrt{n})$-query algorithm 
	for arbitrary hypergraphs. Can we interpolate in between, parameterizing by (non)uniformity? 
	
	Our examples in~\Cref{thm:main} establishing the $\Omega(\sqrt{n})$-query lower bound have $\eavg = \Theta(\sqrt{n})$.
	~\Cref{thm:main-bndmu} gives a $\Ot_\eps(\mu \cdot n^{1/3})$-query algorithm and one subroutine of our algorithm in~\Cref{thm:main}
	achieves an $\Ot_\eps(\mu + \sqrt{n})$-query algorithm. Can we achieve an $\Ot_\eps(\eavg + n^{1/3})$-query algorithm? 
	
	\item {\em Dependence on $\eps$.} The dependence on $1/\eps$ in~\Cref{thm:main-avgdeg} is linear but the dependence on $1/\eps$ in~\Cref{thm:main} is quadratic.
	Can the latter be made linear? Or is the quadratic dependence necessary for estimating number of hyperedges?
	
	\item {\em Power of ``pair'' queries.} A common query in sublinear {\em graph} algorithms
	is a pair query: given two vertices $u, v$, output whether the edge $(u,v)$ exists.
	Beretta, Chakrabarty, and Seshadhri~\cite{BCS26} showed that such queries (in the dual access model) lead
	to $\Ot_\eps(n^{1/4})$-query estimation algorithms. We could imagine generalizing pair queries to hypergraphs,
	wherein the query could be a pair of vertices, and the output could be whether there exists a hyperedge containing the pair,
	or even something stronger, like returning a random hyperedge containing that pair. In that case, our
	lower bounds do not hold. Can we beat the bounds of \Thm{main} and \Thm{main-avgdeg}
	with such queries?
	
	\item {\em Connection with IS queries.} Adar, Hotam, and Levi~\cite{AdHoLe26} recently showed that a hybrid model
	of IS queries with the standard model gives significant improvements for the query complexity
	of edge estimation in {\em graphs}. Could we prove similar results for arbitrary hypergraphs?
	
	\item {\em The degeneracy connection.}  Eden, Ron, and Seshadhri~\cite{ERS19} first observed that the graph
	degeneracy is intimately related to edge estimation. Eden, Ron, and Rosenbaum~\cite{EdRoRo19} made
	this connection even stronger, and Chanda~\cite{Ch26} recently proved the degeneracy bound holds
	for a variety of models. A theory of hypergraph degeneracy has been introduced by Paul-Pena
	and Seshadhri~\cite{PaSe26}. We believe that the study of hyperedge estimation parameterized by
	degeneracy could be a deep research direction.
\end{asparaitem}

\bibliographystyle{plain}
\bibliography{sublinear-edge}

\end{document}

%% file: preamble.tex
\usepackage{a4,geometry}
\usepackage{graphicx}
\usepackage{amsmath,amssymb,amsthm,mathtools}
\usepackage{paralist}
\usepackage{bm}
\usepackage{xspace}
\usepackage{url}
\usepackage{fullpage, prettyref}
\usepackage{boxedminipage}
\usepackage{wrapfig}
\usepackage{ifthen}
\usepackage{mdframed}
\usepackage{color,xcolor}
\usepackage{transparent}
\usepackage{enumitem}
\usepackage{varwidth}
\usepackage{framed}
\usepackage{cancel}

\usepackage{fullpage}
\usepackage{ifthen}
\usepackage{algorithm}
\usepackage[noend]{algpseudocode}
\algrenewcomment[1]{\(\triangleright\) {\small{\emph{\color{blue} #1}}}}
\algnewcommand{\LineComment}[1]{\State \(\triangleright\) \emph{\color{blue} #1}}
\algnewcommand{\Invariant}[1]{\State \(\triangleright\) \emph{\color{red} #1}}
\usepackage[pagebackref,letterpaper=true,colorlinks=true,pdfpagemode=none,urlcolor=blue,linkcolor=blue,citecolor=violet,pdfstartview=FitH]{hyperref}
\usepackage[nameinlink]{cleveref}

\newtheorem{theorem}{Theorem}[section]

\newtheorem{lemma}[theorem]{Lemma}
\newtheorem{claim}[theorem]{Claim}

\newcommand{\ignore}[1]{}

\newcommand{\eq}{\leftarrow}

\newcommand{\cA}{\mathcal{A}}

\newcommand{\cD}{\mathcal{D}}

\newcommand{\calD}{\mathcal{D}}
\newcommand{\calE}{\mathcal{E}}
\newcommand{\calF}{\mathcal{F}}

\newcommand{\RR}{\mathbb{R}}

\newcommand{\h}{{\mathrm{H}\kern1pt}}
\newcommand{\I}{{\mathrm{I}\kern1pt}}

\newcommand{\poly}{\mathrm{poly}}

\newcommand{\ceil}[1]{{\left\lceil{#1}\right\rceil}}

\DeclareMathOperator*{\Var}{\ensuremath{{\mathbf{Var}}}}
\DeclareMathOperator*{\Exp}{\ensuremath{{\mathbf{Exp}}}}
\DeclareMathOperator*{\Prob}{\ensuremath{\mathbf{Pr}}}
\renewcommand{\Pr}{\Prob}
\newcommand{\EX}{\Exp}

\newcommand{\eps}{\ensuremath{\varepsilon}}

\newcommand{\Ot}{\ensuremath{\widetilde{O}}}

\newcommand{\Sec}[1]{\hyperref[sec:#1]{\S\ref*{sec:#1}}} 
\newcommand{\Eqn}[1]{\hyperref[eq:#1]{(\ref*{eq:#1})}} 
\newcommand{\Fig}[1]{\hyperref[fig:#1]{Fig.\,\ref*{fig:#1}}} 
\newcommand{\Tab}[1]{\hyperref[tab:#1]{Tab.\,\ref*{tab:#1}}} 
\newcommand{\Thm}[1]{\hyperref[thm:#1]{Theorem\,\ref*{thm:#1}}} 
\newcommand{\Fact}[1]{\hyperref[fact:#1]{Fact\,\ref*{fact:#1}}} 
\newcommand{\Lem}[1]{\hyperref[lem:#1]{Lemma\,\ref*{lem:#1}}} 
\newcommand{\Prop}[1]{\hyperref[prop:#1]{Prop.~\ref*{prop:#1}}} 
\newcommand{\Cor}[1]{\hyperref[cor:#1]{Corollary~\ref*{cor:#1}}} 
\newcommand{\Conj}[1]{\hyperref[conj:#1]{Conjecture~\ref*{conj:#1}}} 
\newcommand{\Def}[1]{\hyperref[def:#1]{Definition~\ref*{def:#1}}} 
\newcommand{\Alg}[1]{\hyperref[alg:#1]{Alg.~\ref*{alg:#1}}} 
\newcommand{\Clm}[1]{\hyperref[clm:#1]{Claim~\ref*{clm:#1}}} 

\newcommand{\supp}{\mathsf{supp}}

\colorlet{shadecolor}{blue!10}

